%% file: main.tex
\documentclass[a4paper,UKenglish,cleveref,autoref,thm-restate]{lipics-v2021}
\hideLIPIcs
\nolinenumbers
\usepackage{mathtools}
\usepackage{booktabs}
\usepackage{tikz}
\usetikzlibrary{arrows.meta,positioning,calc}
\newcommand{\cN}{\mathcal N}
\newcommand{\cC}{\mathcal C}
\newcommand{\cH}{\mathcal H}
\newcommand{\cB}{\mathcal B}
\newcommand{\cE}{\mathcal E}

\newcommand{\Z}{\mathbb Z}
\newcommand{\ket}[1]{\lvert #1\rangle}
\newcommand{\bra}[1]{\langle #1\rvert}
\newcommand{\ketbra}[2]{\lvert #1\rangle\!\langle #2\rvert}

\newcommand{\SWAP}{\mathrm{SWAP}}
\newcommand{\iSWAP}{\mathrm{iSWAP}}
\newcommand{\CZ}{\mathrm{CZ}}
\DeclareMathOperator{\tr}{tr}
\DeclareMathOperator{\supp}{supp}

\DeclareMathOperator{\girth}{girth}

\newcommand{\Counting}{\textsc{Counting}}
\newcommand{\Election}{\textsc{Election}}
\title{A Quantum Analogue of the Covering Lemma: Limits of Exact Anonymous Computation}
\titlerunning{A Quantum Analogue of the Covering Lemma}
\author{Giuseppe Di Molfetta}{Aix Marseille Univ, CNRS, LIS, Marseille, France
\and Institut Universitaire de France, Paris, France
\and Corresponding author: \href{mailto:giuseppe.dimolfetta@lis-lab.fr}{\nolinkurl{giuseppe.dimolfetta@lis-lab.fr}}}{}{}{}
\author{Petra Camilli}{Universit\'e Paris-Saclay, CNRS, ENS Paris-Saclay, 91190 Gif-sur-Yvette, France}{}{}{}
\authorrunning{G. Di Molfetta and P. Camilli}
\Copyright{Giuseppe Di Molfetta and Petra Camilli}
\EventShortTitle{}
\EventLongTitle{}
\EventAcronym{}
\EventYear{}
\SeriesVolume{}
\ArticleNo{}
\AtBeginDocument{%
  \hypersetup{%
    pdfauthor={Giuseppe Di Molfetta, Petra Camilli},%
    pdfsubject={Research preprint},%
    hidelinks%
  }%
}
\begin{document}
\maketitle
\begin{abstract}
Can quantum communication determine the size of an anonymous network when no size bound is known? Classical impossibility proofs rely on the covering lemma: an execution can be reproduced on a larger network with the same local structure. Quantum interference invalidates this argument—a measurement record possible on a network can become impossible on a cover. We show that the impossibility of exact counting nevertheless survives.
Our main result is a quantum analogue of the covering lemma that preserves local possibilities rather than complete executions. Within a cyclic covering family, every local measurement outcome possible at a node remains possible at a corresponding node in every sufficiently large network. This holds even after information has travelled around the smaller network, beyond the reach of a simple locality argument. We extend the principle to the universal covering tree and finite covers with sufficiently long cycles. The proof combines tensor-network methods with graph coverings: a common representation across network sizes turns normalization into a constraint on which local outcomes can occur.
For uniform synchronous protocols with independently prepared nodes, this support-lifting principle rules out zero-error counting and even the production of a guaranteed finite size bound on anonymous rings and the other cyclic covering families considered here. It also rules out exact leader election with termination in finite time with probability one on these families. Measurements, classical control, and local memory growing with time are allowed; no predetermined running-time bound is assumed. Quantum interference thus defeats execution lifting without enabling infallible certification of network size.
\end{abstract}
\clearpage
\input{stacs/sections/01-introduction}
\input{stacs/sections/02-model}
\input{stacs/sections/03-covering}
\input{stacs/sections/04-support}

\input{stacs/sections/05-universal}
\input{stacs/sections/06-consequences}
\input{stacs/sections/07-discussion}

\label{end:body}
\clearpage

\input{main.bbl}
\clearpage
\appendix
\input{stacs/sections/appendix}

\end{document}

%% file: stacs/sections/01-introduction.tex
\section{Introduction}\label{sec:intro}

An anonymous network consists of processors without unique identifiers. Nodes with the same degree and input use the same local program and initialization, although their states may diverge during an execution. A central tool for proving impossibility in such networks is the covering lemma. Given a finite classical execution of positive probability on a network, repeating the local choices at corresponding nodes produces an execution of positive probability on every covering network with the lifted data. Local decisions are therefore repeated across the copies: an output reporting the smaller network’s size becomes incorrect on a larger cover, while a unique leader is replicated once per sheet. This argument underlies classical impossibility results for anonymous computation~\cite{angluin1980local,itai1990symmetry,yamashita1996computing}.

Quantum communication changes the picture, at least when something is known about the size. Exact quantum leader election~\cite{tani2012exact} and solitude verification~\cite{tani2017solitude} need only a known upper bound on the number of nodes, and the bound may be arbitrarily loose. Classical zero-error randomized election  on rings is possible with a tight size bound $n\le N<2n$ ~\cite[Sections 1.2 and 2.3]{aharonov2017dining}, while deterministic election remains impossible on symmetric rings even when $n$ is known; what quantum protocols buy is the looseness of the bound, not the ability to use one. Aharonov, Ganz and Magnin ask whether any size bound is needed at all: can an exact quantum protocol determine the size of an anonymous ring, or just an upper bound on it, with no prior size information~\cite{aharonov2017dining}? {We answer this question negatively for protocols that are uniform under lifts, start from independently prepared nodes, and use finite local resources at every fixed horizon. Measurements, classical messages, and unbounded running time are allowed.} The networks we compare are cyclic covers: for a base network $\cN$, the $m$-sheet cyclic cover $\cN_m$ consists of $m$ copies of $\cN$ in which some links lead to the next copy around instead of closing inside their own, so that the rings $C_{nm}$ are the cyclic covers of $C_n$. We assume the tower $(\cN_m)_{m\geq1}$ connected and the local dynamics lifted uniformly along it, and nothing about correctness. A node's local quantum state is described by a density matrix. Its \emph{support} is the subspace spanned by eigenvectors with nonzero eigenvalues. Support inclusion has an operational meaning: 
if $\operatorname{supp}\rho \subseteq \operatorname{supp}\sigma$,
every local measurement outcome that has positive probability
in state $\rho$ also has positive probability in state $\sigma$,
for the same measurement. The probabilities need not be equal. 

Writing $\rho_m(v,s)$ for the state of a node above $v$ after $s$ rounds on $\cN_m$, our main theorem is the inclusion
\begin{equation}\label{eq:intro-support}
 \supp\rho_k(v,s)\subseteq\supp\rho_m(v,s)
 \qquad(k\geq1,\;m\geq2s+1).
\end{equation}
{A possible local outcome---one having positive probability---cannot certify an exact size: it remains possible on arbitrarily large networks. All sufficiently large tower members have the same local state (\cref{lem:injrad}), and every member's support lies in that common support (\cref{thm:supportlift}), even after information has travelled around the smaller network (\cref{fig:overview}).}

{When every tower member is an admissible instance, no zero-error protocol can output its size or a finite upper bound, and total exact computation of a function forces constancy along the tower. Exact leader election is also excluded, by a leader-to-counting reduction. These conclusions concern tower members: they do not exclude partial certificates on other graphs or partial zero-error certificates for eventually constant functions.}

{The support inclusion extends to the universal covering tree and every finite cover of girth greater than $2s$. An explicit two-round protocol has orthogonal local states on $C_2$ and $C_4$, showing that the threshold cannot be weakened uniformly from $2s+1$ to $2s$.}

\input{stacs/figures/overview}
The need for a new argument is seen in a two-round protocol whose global measurement record has positive probability on{ $C_3$ and zero probability on its cover $C_6$. The cancellation comes from interference between lifted amplitude contributions. Section 3 gives the construction and its exact probability formula. }

{Our proof uses tensor-network methods~\cite{fannes1992finitely,perezgarcia2007mps}. A uniform finite-horizon dilation gives one matrix-product tensor for all sheet counts. Normalization constrains its transfer spectrum; an irreducible-block reduction then supplies positive-definite boundary operators and the support inclusion. Locality gives the explicit threshold in~\eqref{eq:intro-support}. To reach the universal cover, we factor finite homology covers into cyclic covers until the required girth is reached.}

We formulate the model as a network of processors made of finite quantum cells, in the spirit of quantum cellular automata~\cite{arrighi2019overview,farrelly2020review}; graph-based quantum-walk models of distributed computation are close relatives~\cite{roget2024scheme,roget2026model}. Non-signaling lower bounds~\cite{coiteuxroy2024coloring} are another causal framework, but they compare inputs on the same vertex set, allow pre-shared quantum states, and do not give the cross-size support statement proved here; conversely, our hypothesis of independent preparation is stronger than the mere absence of entanglement.

%% file: stacs/figures/overview.tex
\begin{figure}[t]
\centering
\begin{tikzpicture}[scale=0.85,every node/.style={font=\footnotesize,transform shape},
  nd/.style={circle,draw,thick,fill=white,inner sep=0pt,minimum size=2.1mm},
  vn/.style={circle,draw,thick,fill=black,inner sep=0pt,minimum size=2.1mm},
  cone/.style={circle,draw=none,fill=black!14,inner sep=0pt,minimum size=5.2mm},
  arr/.style={-{Latex[length=2mm]},thick},
  bad/.style={-{Latex[length=2mm]},thick,densely dashed}]
 \begin{scope}
  \foreach \i in {0,1,2} { \coordinate (p\i) at ({90+120*\i}:0.62); }
  \foreach \i in {0,1,2} { \node[cone] at (p\i) {}; }
  \draw[thick] (p0)--(p1)--(p2)--cycle;
  \node[vn] at (p0) {}; \node[nd] at (p1) {}; \node[nd] at (p2) {};
  \node[above=1pt] at (p0) {$v$};
  \node[align=center] at (0,-1.25) {{$C_3$: whole ring in the cone}};
  \begin{scope}[shift={(4.3,0)}]
   \foreach \i in {0,...,6} { \coordinate (q\i) at ({90+360/7*\i}:0.95); }
   \foreach \i in {0,1,2,5,6} { \node[cone] at (q\i) {}; }
   \draw[thick] (q0)--(q1)--(q2)--(q3)--(q4)--(q5)--(q6)--cycle;
   \node[vn] at (q0) {}; \foreach \i in {1,...,6} { \node[nd] at (q\i) {}; }
   \node[above=1pt] at (q0) {$v$};
   \node[align=center] at (0,-1.25) {{$C_m$, $m\ge2s+1$}\\{stabilized local state}};
  \end{scope}
  \draw[arr] (0.95,0.05) -- (3.0,0.05);
  \node[align=center] at (2.0,0.65) {$\supp\rho_3(v,2)$\\$\subseteq\supp\rho_m(v,2)$};
  \node at (2.15,-1.85) {(a) support lifting along a tower, $s=2$};
 \end{scope}
 \begin{scope}[shift={(9.7,0.75)}]
  \node[anchor=east] at (0,0) {classical execution};
  \draw[arr] (0.1,0) -- node[above,inner sep=1pt] {$p>0$} (1.5,0);
  \node[anchor=west] at (1.6,0) {every cover};
  \node[anchor=east] at (0,-0.8) {quantum record on $C_3$};
  \draw[bad] (0.1,-0.8) -- node[above,inner sep=1pt] {$p=0$} (1.5,-0.8);
  \node[anchor=west] at (1.6,-0.8) {$C_6$};
  \node[anchor=east] at (0,-1.6) {local outcome on $\cN_k$};
  \draw[arr] (0.1,-1.6) -- node[above,inner sep=1pt] {$p>0$} (1.5,-1.6);
  \node[anchor=west] at (1.6,-1.6) {$\cN_m$, $m\ge2s+1$};
  \node at (0.6,-2.6) {(b) what lifts to covers and what does not};
 \end{scope}
\end{tikzpicture}
\caption{{(a) The tower $\cN_m=C_m$ over the one-node loop graph, with $s=2$. The source cone contains the whole $C_3$. For $m\ge2s+1$, the relevant causal data agree with those of the infinite path: only links incident to $B(v,s-1)$ matter. Every local outcome possible on $C_3$ remains possible on those $C_m$ (\cref{thm:supportlift}). (b) A finite classical execution lifts to every cover. A quantum global record can vanish (\cref{thm:nolift}: positive probability on $C_3$, zero on $C_6$); local support lifts beyond the threshold.}}
\label{fig:overview}
\end{figure}
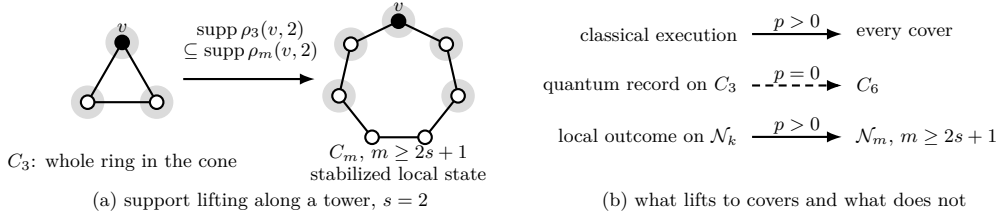

%% file: stacs/sections/02-model.tex
\section{Model and finite-time locality}\label{sec:model}
{We consider anonymous nodes with quantum registers and classical control. They are prepared independently, but synchronous rounds of local operations and link interactions can create entanglement. The model separates internal processing from communication. We define the network, local dynamics, and correctness requirements, then establish finite-time locality and introduce graph coverings.}
\subsection{Port graphs and local processors}
\begin{definition}[Port graph]\label{def:portgraph}
A finite \emph{port graph} is $\cN=(V,A,\iota,\nu)$, where $V$ is the node set, $A$ the port set, $\iota:A\to A$ a fixed-point-free involution, and $\nu:A\to V$ the owner map. The ports at $v$ are $A_v=\nu^{-1}(v)$; its degree is $|A_v|$. A link is a pair $\{a,\iota(a)\}$.
\end{definition}
All port graphs are connected. Loops and parallel links are allowed, with their ports kept distinct, and so is the edgeless one-node graph. Distances are those of the underlying multigraph, and $B(W,r)$ is the node ball of radius $r$ around $W$. The ring $C_n$ has nodes $\Z_n$, ports $(v,+),(v,-)$, and $\iota(v,+)=(v+1,-)$; thus $C_1$ is one node with a loop and $C_2$ two nodes joined by parallel links. These small cases are convenient for the cyclic tower; the counting obstruction applies just as well to classes containing only simple rings of length at least three.
A degree-$d$ processor consists of finitely many quantum cells, $d$ of them designated as port cells, and a bijection between the incident ports and the port cells, called an \emph{attachment}. Each cell has a finite-dimensional Hilbert space with a computational basis, and all port cells have a common dimension $q$. A node step applies a unitary to the cells of the node; the lower bounds below hold for arbitrary node unitaries, and we do not restrict how they are compiled. In the explicit protocols the cells are qubits on a path and the node unitaries are circuits of one-cell gates and two-cell gates between adjacent cells, which generate every unitary of the node once swaps are used for routing\footnote{Connectivity alone does not guarantee universality for arbitrary cell dimensions. For example, on a three-cell path with dimensions $2,1,2$, gates acting on adjacent cells cannot couple the two endpoint qubits. Our impossibility proofs do not require internal universality: they allow arbitrary operations within each node.}. A port cell also carries a classical message slot. The link operation is a fixed two-port-cell unitary $E$ with $E\SWAP=\SWAP E$, together with the swap of the classical slots; the symmetry makes the link independent of an orientation. The examples use $E=\SWAP$, and the proofs allow any such $E$.

{Fresh internal quantum registers may be initialized locally in fixed pure states during execution, subject to the uniform finite-horizon resource bound below. Their allocation is part of the uniform local program.}

\begin{definition}[Round and algorithm]\label{def:algorithm}
Each node has a countable classical state space $K$ containing its registers, message slots and output register. In state $k$ at round $t$, it applies a local unitary $G(k,t)$, measures finitely many cells in the computational basis with outcome $r$, and updates its classical state by a deterministic map $\delta(k,t,r)$. The node channel is
\begin{equation}\label{eq:node-channel}
 \Phi_{v,t}(\rho)=\sum_{k,r}K_{k,r}\rho K_{k,r}^{\dagger},
 \qquad K_{k,r}=P_rG(k,t)\otimes\ketbra{\delta(k,t,r)}{k},
\end{equation}
where $P_r$ is the measurement projector. A round applies the tensor product of the node channels and then the link layer $\cE$:
\[
 \Phi_t=\cE\circ\bigotimes_v\Phi_{v,t},
 \qquad \rho_t=\Phi_t\circ\cdots\circ\Phi_1(\rho_0).
\]
A \emph{uniform algorithm} assigns the processor, local program, initial classical state and local quantum state as functions of the degree and of the input symbol alone, the input alphabet being finite. The local rules thus define dynamics for every attachment, on lifted instances outside a task promise as much as inside; correctness is a separate requirement. The initial state is
\begin{equation}\label{eq:product}
 \rho_0=\bigotimes_{v\in V}\ketbra{\psi_{\deg(v),x(v)}}{\psi_{\deg(v),x(v)}}.
\end{equation}
Classical initial states are part of these local factors. Entanglement \emph{within} a node is allowed; correlations \emph{between} nodes are not.
\end{definition}
Each node’s output register is initially set to $\bot$, meaning “undecided.” A node may replace $\bot$ with an output value, after which that value cannot change. For a function task, a \emph{zero-error} algorithm never outputs a wrong value, with probability one on every admissible instance and every attachment; it \emph{solves the task exactly} if, in addition, every node outputs after a finite time with probability one, a deterministic bound on that time not being required. Nodes may keep forwarding auxiliary messages after they have output. As in the LOCAL model~\cite{linial1992locality,peleg2000distributed}, classical computation and classical message length are free. 

The lower bounds rest on three hypotheses: uniformity under lifts, product preparation, and \emph{finite local resources at every fixed horizon}. The last means that for every base node type and every horizon $s$, the total Hilbert space dimension available at a node, together with the sets of reachable classical states and messages, is bounded by a constant that does not depend on the cover under consideration. These bounds may grow with $s$, allowing larger memories and messages in later rounds. However, at any fixed horizon, they must be independent of the number of sheets. With fixed cells the condition holds automatically, by finite branching. The hypotheses do not cover an initially shared correlated state, even a separable one, nor an oracle supplying a size-dependent parameter, which is input information and need not lift identically along the tower.

\subsection{Light cones and local lifting}\label{sec:locality}
Write $X_W$ for all cells and classical registers at nodes in $W$, and $\mathcal A_{X_W}$ for the operators supported there. At a fixed finite horizon the reachable spaces are finite-dimensional, uniformly over the covers considered (\cref{lem:purification}).

\begin{lemma}[Light cone]\label{thm:cone}
After $s$ rounds, the expectation of an observable supported at $W$ depends only on the local data of the nodes in $B(W,s)$ and on the links with at least one endpoint in $B(W,s-1)$; for $s=0$ no link is involved. In particular, the dual channel sends $\mathcal A_{X_W}$ into $\mathcal A_{X_{B(W,s)}}$.
\end{lemma}
\begin{proof}
Follow an observable backwards. In a link layer, the gates disjoint from its support commute with it and cancel, and the remaining gates enlarge its node support by at most one graph step; the dual of a node channel is unital, so node operations outside the current support act on the identity and do not enlarge it. After $s$ rounds the radius is $s$. In the last link layer followed backwards only the links incident to $B(W,s-1)$ occur; a link joining two boundary nodes at distance $s$ is never needed. Taking the expectation in the product state~\eqref{eq:product} uses only the initial states of the nodes in $B(W,s)$.
\end{proof}

\begin{definition}[Covering and lifted data]\label{def:covering}
A \emph{covering} $\varphi:\widetilde\cN\to\cN$ consists of surjections on nodes and on ports that commute with the owner maps and the involutions and restrict to a bijection on the ports at each node. The lift of an instance copies its inputs, processors, initial states and attachments through these maps. The \emph{injectivity radius} at $\tilde v$ is the largest $r$ for which the node map is injective on $B(\tilde v,r)$, or $\infty$ if there is no largest finite $r$.
\end{definition}

\begin{proposition}[Local lifting]\label{prop:locallift}
If a covering is injective on $B(\tilde v,s)$, the local states at $\tilde v$ and at $\varphi(\tilde v)$ after $s$ rounds are equal under the natural identification of their local systems.
\end{proposition}
\begin{proof}
A covering maps node balls onto node balls, so injectivity identifies the two balls of radius $s$ bijectively, and every link incident to the ball of radius $s-1$ has a unique lift there, with its other endpoint in the ball of radius $s$. The causal data of \cref{thm:cone}, initial product factors and attachments included, therefore agree, and so do all local expectations and the reduced states. For $s=0$ this is just local preparation.
\end{proof}
The argument then applies to a locally finite infinite covering graph: a local state at a finite time is defined by its finite causal cone, without any global infinite tensor-product state.

%% file: stacs/sections/03-covering.tex
\section{Failure of global record lifting}\label{sec:coverings}

The classical covering argument carries every finite record of positive probability from a network to each of its covers. Here is a quantum protocol for which this fails, while \cref{prop:locallift} of course still holds.

\paragraph*{A two-round protocol.}
Each ring node has port qubits $p_1,p_2$ and internal qubits $d_1,d_2$, arranged along the path $p_1-d_1-d_2-p_2$ and all initially in $\ket0$. In the node step of round~1 the node prepares
\[
 \ket\Phi_{d_1p_1}\otimes\ket\Phi_{d_2p_2},
 \qquad \ket\Phi=(\ket{00}+\ket{11})/\sqrt2,
\]
and the link step swaps the two port cells of each link. In round~2 each node tests the rank-one projector
\begin{equation}\label{eq:record-projector}
 \Pi=\ketbra{a}{a}_{p_1d_2}\otimes\ketbra{a}{a}_{p_2d_1},
 \qquad \ket a=(\ket{00}+i\ket{11})/\sqrt2,
\end{equation}
and stores the result: it rotates $\ket a$ to $\ket{00}$ on each pair, measures the four qubits, and records whether it saw $0000$. The pairs are not adjacent on the path, but they can be routed inside the processor. Exchanging $p_1$ with $p_2$ and $d_1$ with $d_2$ simultaneously leaves preparation and test unchanged, so the protocol depends neither on a consistent orientation of the ring nor on a particular attachment. The second link step comes after the outcomes have been stored and does not affect them.

\begin{theorem}[Failure of global record lifting]\label{thm:nolift}
For this protocol, the probability that every node of $C_n$ records $\Pi$ is
\begin{equation}\label{eq:probability}
 p_n=\frac{|1+(-i)^n|^4}{16^n}.
\end{equation}
It vanishes exactly when $n\equiv2\pmod4$. In particular, the all-$\Pi$ record has positive probability on $C_3$ and probability zero on its double cover $C_6$.
\end{theorem}
\begin{proof}
Number the nodes modulo $n$ and attach $p_2$ of node $j$ to $p_1$ of node $j+1$, which the symmetry just noted allows. After the first link step there is a Bell pair on $(d_2\text{ at }j,\,p_1\text{ at }j+1)$ and another on $(p_2\text{ at }j,\,d_1\text{ at }j+1)$ for every $j$.

The projections~\eqref{eq:record-projector} close these pairs into two disjoint chains, one through all the $d_2,p_1$ qubits and one through all the $d_1,p_2$ qubits. Along one chain the Bell tensors identify consecutive binary indices, so with $D=\operatorname{diag}(1,-i)$ the chain contracts to
\[
 2^{-n}\tr(D^n)=2^{-n}\bigl(1+(-i)^n\bigr),
\]
the factor $2^{-n}$ collecting $n$ normalizations $1/\sqrt2$ from the Bell states and $n$ from the projections. The other chain contributes the same factor. The total amplitude is $4^{-n}(1+(-i)^n)^2$, and its modulus squared is~\eqref{eq:probability}.
\end{proof}

In this expansion each chain carries one constant binary label, so there are exactly four nonzero contributions, indexed by $a,b\in\{0,1\}$, with amplitudes
\begin{equation}\label{eq:four-histories}
 4^{-n}(-i)^{n(a+b)}.
\end{equation}
On $C_3$ their phases are $1,i,i,-1$ and add up to $2i$. On $C_6$ they are $1,-1,-1,1$ and add up to zero. All four contributions on $C_6$ lift from $C_3$: the cancellation takes place \emph{among lifted contributions}, not between lifted and non-lifted histories. What matters is that complex amplitudes are being added instead of nonnegative probabilities. This is a global event, no single node is assumed to see the whole record. Note also that \cref{thm:nolift} neither proves nor disproves anything about local supports beyond the injectivity radius; that needs a different argument, the subject of the next section.

%% file: stacs/sections/04-support.tex
\section{Support lifting along cyclic towers}\label{sec:support}
We now establish the quantum replacement for classical execution lifting. Rather than reproducing complete executions, we show that every possible local observation remains possible on sufficiently large networks in the same cyclic covering family. The first step is to represent measurements and classical control by unitary evolution on enlarged local systems, from which the original protocol can be recovered. Crucially, this representation must be chosen uniformly across the family: the local construction cannot depend on the network size. We then combine a common tensor-network representation with finite-time locality to prove the support-lifting theorem.
\subsection{Uniform finite-horizon dilation}\label{sec:purification}
\begin{lemma}[Uniform finite-horizon dilation on a family of covers]\label{lem:purification}
Let $\mathfrak F$ be a family of finite covers $\pi_H:H\to\cN$ of a base instance, with lifted local data, and let $s$ be a horizon. Consider a uniform, product-prepared protocol whose dynamics is defined up to time $s$ on every $H\in\mathfrak F$ and satisfies the finite-resource condition up to time $s$ uniformly over $\mathfrak F$. Then there is a closed protocol, with unitary node steps on enlarged node spaces, a pure product initial state and the same uniformity, whose node and port dimensions do not depend on $H\in\mathfrak F$, together with local channels $R_{v,t}$ ($t\le s$), the same at every node above $v$ in every $H$, such that
\[
 \rho_H(w,t)=R_{v,t}\!\left(\widehat\rho_H(w,t)\right)
 \qquad\bigl(H\in\mathfrak F,\ w\in\pi_H^{-1}(v)\bigr),
\]
where $\widehat\rho_H$ denotes the local state of the closed protocol. Nothing is assumed about covers outside $\mathfrak F$, and no correctness on any task is used. The dilation preserves the communication pattern and the original quantum link operation; classical message swaps are represented as swaps of additional basis registers.
\end{lemma}
\begin{proof}
The resource condition provides finite-dimensional active node and port spaces, together with  finite sets containing classical states and messages reachable by time $s$ rounds, common to all members of $\mathfrak F$. (With fixed cells these come from finite branching: the base has finitely many node types and bounded degree, each classical state selects a measurement with finitely many outcomes, the transition map is deterministic, and received messages are tuples from the finite sets of the previous round; induction over the rounds gives the supersets. Finiteness on each cover separately would not give a uniform bound.) We complete the local operations on the unused padding subspaces to trace-preserving maps, identically at corresponding nodes.

Purify each initial node state locally, identically at corresponding nodes, and replace each node channel~\eqref{eq:node-channel} by the isometry
\[
 V_{v,t}\ket\psi=\sum_{k,r}K_{k,r}\ket\psi\otimes\ket{k,r}_{\mathrm{env}(v,t)},
\]
which satisfies $V_{v,t}^\dagger V_{v,t}=I$ by completeness and extends to a unitary with local pure ancillas, again identically at corresponding nodes. Classical messages become basis registers swapped on the links, and later gates may be controlled by the \emph{active} classical registers. The environment $\mathrm{env}(v,t)$ stays at node $v$ and is never touched again. Tracing it out recovers the node channel of round $t$, and tracing out the older environments commutes with the node step and with the link layer, so
\[
 \tr_{\mathsf E}\widehat\rho_{H,t+1}
 =
 \Phi_{H,t+1}\!\left(
   \tr_{\mathsf E}\widehat\rho_{H,t}
 \right),
 \qquad 0\leq t<s,
\]
where $\mathsf E$ comprises all environment and purification registers preallocated for the fixed horizon $s$, including those not yet used. Unused registers remain in their fixed pure initial states.
By induction the original joint state is recovered after every round, even when the active registers are correlated with old environments. The partial trace over the environments of node $v$, composed if one wishes with the dephasing of its retained classical registers, is $R_{v,t}$. All ancillas are preallocated in local pure states and retained throughout the closed evolution through round $s$. Ancillary registers are traced out only when recovering the original states. Every choice depends on the base, the protocol, $\mathfrak F$ and $s$, but not on the member $H$: it is one dilation for the whole family, living in enlarged node spaces and claiming nothing about compilation on the original cells.
\end{proof}

For positive operators in finite dimension,
\begin{equation}\label{eq:support-order}
 \supp\sigma\subseteq\supp\tau
 \quad\Longleftrightarrow\quad
 \sigma\leq C\tau\ \text{for some finite }C>0,
\end{equation}
the forward implication by restricting to the support of $\tau$ and bounding $\tau^{-1/2}\sigma\tau^{-1/2}$, the converse by comparing kernels. Positive maps preserve the inequality, so a support inclusion between the dilated local states survives the channel of \cref{lem:purification}, partial trace and dephasing included, and we may work with closed protocols from now on.

\subsection{Cyclic covers and stabilization of local states}
\begin{definition}[Cyclic tower]\label{def:cycliccover}
Let $c:A\to\{-1,0,1\}$ satisfy $c(\iota(a))=-c(a)$. The $m$-sheet cyclic cover $\cN_m$ has nodes $V\times\Z_m$, ports $A\times\Z_m$, and
\[
 \nu_m(a,j)=(\nu(a),j),\qquad
 \iota_m(a,j)=(\iota(a),j+c(a)).
\]
Projection onto the first coordinate is a covering, and $\cN_1=\cN$. Each set $V\times\{j\}$ is a sheet. The tower is \emph{connected} if every $\cN_m$ is connected; all local data are lifted from $\cN$. For a task class $\cC$, a connected tower is \emph{$\cC$-admissible} if every lifted instance $\cN_m$ belongs to $\cC$. Connectedness is a geometric condition; admissibility is needed only when correctness is invoked.
\end{definition}
{Members of the tower share a common base but need not cover one another.}
A link of voltage zero stays inside its sheet; a link of nonzero voltage joins consecutive sheets. Giving voltage $1$ to one orientation of a non-bridge link, $-1$ to the reverse port and $0$ to everything else produces a connected tower; the rings $C_{nm}$ arise from $C_n$ in this way, and a torus can be cut along the links between two consecutive coordinate hyperplanes.

\begin{lemma}[Local stabilization]\label{lem:injrad}
For $m'\geq1$, the covering $\cN_{mm'}\to\cN_m$ obtained by reducing sheet indices modulo $m$ has injectivity radius at least $\lfloor(m-1)/2\rfloor$. For a fixed base node $v$ and horizon $s$, the reduced states $\rho_m(v,s)$ therefore have a common value $\rho_\infty(v,s)$ for all $m\geq2s+1$. For every $m$, the state is also independent of the chosen sheet above $v$.
\end{lemma}
\begin{proof}
A path of length $\ell$ changes the integer sheet coordinate by at most $\ell$. Two distinct lifts of a node with the same image modulo $m$ differ by a nonzero multiple of $m$ in that coordinate, so they are at distance at least $m$, and a ball of radius $r$ with $2r<m$ contains no such pair. For $m,m'\geq2s+1$, compare the two local states through $\cN_{mm'}$ with \cref{prop:locallift}. Finally, translating the sheets preserves the lifted instance and the product initial state, hence the evolved state.
\end{proof}

\subsection{A common matrix-product representation}
Fix a closed protocol and a horizon $s$. Let $\cH_{\mathrm{sheet}}$ be the Hilbert space of one entire sheet, dilation cells included, with basis $\{\ket x\}$, and write $\ket{\Psi_m}$ for the state on $\cN_m$ after $s$ rounds. The boundary indices record the finitely many terms introduced when gates crossing between sheets are decomposed. We group the circuit by sheets. Decomposing the gates between neighboring sheets leaves two auxiliary indices for each sheet, one for each boundary. For a fixed final basis label $x$, the resulting coefficients form a matrix $A_x$. Multiplying these matrices and taking a trace reconnects the sheets into a cycle. {The representation is exact; no truncation or approximation is used.}
\begin{lemma}[Common matrix-product tensor]\label{lem:mps}
There are a finite-dimensional bond space $\cB$ and matrices $A_x$ on $\cB$, independent of $m$, such that
\begin{equation}\label{eq:mps}
 \ket{\Psi_m}=\sum_{x_1,\ldots,x_m}
 \tr(A_{x_1}\cdots A_{x_m})\ket{x_1,\ldots,x_m}
 \qquad(m\geq1).
\end{equation}
\end{lemma}
\begin{proof}
Draw the $s$-round circuit on the cyclic sheets. The gates that cross a sheet boundary are link operations $E$; expand each in an operator-Schmidt decomposition
\[
 E=\sum_{\alpha=1}^{R}L_\alpha\otimes R_\alpha,
 \qquad R\leq q^2,
\]
where $q$ now stands for the effective port dimension, message registers included, and absorb one factor into each incident sheet, in the order dictated by the sign of the voltage. If $L$ links cross each boundary, the joint boundary index has dimension at most $q^{2sL}$, or a product of the relevant finite dimensions when they vary with time. Inside a sheet, contract all gates and the initial sheet state; what remains, the final physical index $x$ of the sheet and the two boundary indices, is $A_x$. Since the local data are lifted uniformly, this tensor is the same on every sheet and does not depend on $m$. Reconnecting the boundary indices around the cycle gives the trace in~\eqref{eq:mps}, also for $m=1$, where the two ends are reconnected to the same sheet and the original circuit is recovered; in the loop and parallel-link cases distinct half-edges remain distinct cells, so no quantum subsystems get identified.
\end{proof}

The transfer map is the completely positive map
\begin{equation}\label{eq:transfer}
 \mathbb T(Z)=\sum_x A_xZA_x^\dagger,
 \qquad\text{with matrix representation}\qquad \mathbb T=\sum_x A_x\otimes\overline{A_x},
\end{equation}
and its ordinary matrix trace satisfies $\tr(\mathbb T^m)=\langle\Psi_m|\Psi_m\rangle$.
{Here we identify the map with its matrix under the vectorization
\[
 |Z\rangle\!\rangle=\sum_{i,j}Z_{ij}|i\rangle\otimes|j\rangle,
 \qquad
 |\mathbb T(Z)\rangle\!\rangle
 =\left(\sum_x A_x\otimes\overline{A_x}\right)|Z\rangle\!\rangle.
\]
The bar denotes entrywise complex conjugation.}

{The next lemma uses the trace-of-powers argument of Cirac et al.~\cite[Proposition III.1]{cirac2017mpu}, in their matrix-product-unitary setting. We include a self-contained proof for the state tensor of \cref{lem:mps}.}

\begin{lemma}[Transfer spectrum]\label{lem:spectrum}
The only nonzero eigenvalue of $\mathbb T$ is $1$, with algebraic multiplicity one. There are a rank-one spectral projector $P$ and a nilpotent map $N$ such that $PN=NP=0$ and
\[
 \mathbb T^j=P+N^j\qquad(j\geq1).
\]
In particular $\mathbb T^j=P$ for all large $j$, although $\mathbb T$ itself need not have rank one.
\end{lemma}
\begin{proof}
The circuit is unitary and starts in a normalized state for every $m$, so $\tr(\mathbb T^m)=1$ for all $m\geq1$; this uses only that the local dynamics are trace preserving, not correctness or termination on a task class. If $\lambda_i$ are the distinct nonzero eigenvalues, with algebraic multiplicities $\mu_i$, then $\sum_i\mu_i\lambda_i^m=1$ for all $m$. Add the scalar $1$ to the list of bases if it is not there and move $1^m$ to the left. The sequences $m\mapsto\lambda^m$ for distinct nonzero $\lambda$ are linearly independent (a Vandermonde argument), so every coefficient vanishes, which is the claim about the spectrum. The Jordan decomposition provides $P$ and $N$.
\end{proof}

\subsection{The support theorem}
\begin{theorem}[Support lifting to cyclic covers]\label{thm:supportlift}
Let $(\cN_m)_{m\geq1}$ be a connected cyclic tower with lifted local data. Consider a local protocol uniform under lifting, initialized in a product state across nodes, with finite local resources at every fixed horizon uniformly in the number of sheets. Assume its trace-preserving dynamics are defined on every tower member. Then, for every base node $v$, every finite horizon $s$, and under the natural identification of local systems,
\begin{equation}\label{eq:supportlift}
 \supp\rho_k(v,s)\subseteq\supp\rho_m(v,s)
 \qquad\text{for all }k\geq1\text{ and }m\geq2s+1.
\end{equation}
Equivalently, every local measurement effect with positive probability on $\cN_k$ has positive probability at the corresponding node of every such $\cN_m$. No correctness on any problem class is assumed.
\end{theorem}
\begin{proof}
Apply \cref{lem:purification} to the family $\{\cN_j\to\cN:j\ge1\}$: one closed protocol for the whole tower, and by~\eqref{eq:support-order} and the common recovery channel it suffices to prove the inclusion for that closed protocol, which we assume from now on. Fix $k$ and let $\rho_\infty$ denote the common local state of \cref{lem:injrad}.

{We compare the small closed network with a region inside a much larger network. After an irreducible-block reduction, the exterior supplies positive-definite boundary operators, yielding the support inclusion. Locality then gives the explicit threshold.}

\emph{1. Reduction to an irreducible block.}
Choose a maximal flag of subspaces of $\cB$ invariant under all $A_x$. In a basis adapted to it, the matrices are block upper triangular and the diagonal families $A_x^{(i)}$ are irreducible. Products remain block upper triangular, so
\[
 \Psi_m=\sum_i\Psi_m^{(i)},\qquad
 \Psi_m^{(i)}(x_1,\ldots,x_m)=\tr(A_{x_1}^{(i)}\cdots A_{x_m}^{(i)}).
\]
In the basis indexed by pairs of blocks, ordered lexicographically, the transfer matrix is block upper triangular with diagonal blocks $\sum_xA_x^{(i)}\otimes\overline{A_x^{(j)}}$. Every eigenvalue of every such block is therefore an eigenvalue of $\mathbb T$, counted with algebraic multiplicity. Let $\mathbb T_i$ denote the block with $i=j$. If $\mathbb T_i$ is nilpotent, then
\[
 \|\Psi_m^{(i)}\|^2=\tr(\mathbb T_i^m)=0\qquad(m\geq1).
\]
Not all $\mathbb T_i$ can be nilpotent, because $\Psi_m$ is normalized, and by \cref{lem:spectrum} exactly one of them carries the simple eigenvalue $1$, the others being nilpotent. So exactly one irreducible diagonal family contributes to every $\Psi_m$; we replace $A_x$ and $\mathbb T$ by this family and its transfer map, and~\eqref{eq:mps} remains valid for every $m$.

\emph{2. The boundary operators are positive definite.}
By \cref{lem:spectrum} there is a positive integer $D$ with $\mathbb T^j=P$ for all $j\geq D$, where $P$ has rank one as a linear map on matrices and, being a power of $\mathbb T$, is completely positive. Choose $X\geq0$ with $r=P(X)\ne0$, which is possible because the positive matrices span the matrix space and $P\ne0$. Since the range of $P$ is one-dimensional and $P$ is positive, $P$ is $r$ times a positive functional, represented by some $l\geq0$:
\begin{equation}\label{eq:caps}
 P(Z)=r\tr(lZ),\qquad \tr(lr)=1.
\end{equation}
The normalization comes from $P^2=P$, and $\mathbb TP=P\mathbb T=P$ gives $\mathbb T(r)=r$ and $\mathbb T^*(l)=l$.

Now let $z\in\ker r$. Each term of $r=\sum_xA_xrA_x^\dagger$ is positive, so $r^{1/2}A_x^\dagger z=0$ for every $x$: $\supp r$ is invariant under all $A_x$, and irreducibility with $r\ne0$ forces $r>0$. The adjoint family is irreducible too, since a subspace invariant under all $A_x^\dagger$ has an orthogonal complement invariant under all $A_x$, and the same argument applied to $l=\sum_xA_x^\dagger lA_x$ gives $l>0$.

\emph{3. Full-rank boundaries contain every closed-chain amplitude.}
{As illustrated in \cref{fig:mps}}, for a configuration $x=(x_1,\ldots,x_k)$ of $k$ sheets put $M_x=A_{x_1}\cdots A_{x_k}$, so that the closed-chain state is $\Psi_k(x)=\tr(M_x)$. Choose $M\geq k+D$ and reduce $\ket{\Psi_M}$ to $k$ consecutive sheets, sheet zero among them. Let $y=(y_1,\ldots,y_{M-k})$ denote a configuration of the
remaining sheets, and put $M_y=A_{y_1}\cdots A_{y_{M-k}}$.
Taking the partial trace over these sheets gives
\[
  \sigma_M(x,x')
  =
  \sum_y
  \operatorname{tr}(M_x M_y)\,
  \overline{\operatorname{tr}(M_{x'} M_y)}.
\]
The sum over exterior configurations is represented by the
transfer matrix $\mathbb T^{M-k}$. Since $M-k\geq D$, this
power equals $P$. Using $P(Z)=r\,\operatorname{tr}(lZ)$,
the reduced density matrix therefore has entries
\begin{equation}\label{eq:gram}
 \sigma_M(x,x')=\tr\bigl(lM_xrM_{x'}^\dagger\bigr).
\end{equation}
Factor $r=\sum_a\ketbra{r_a}{r_a}$ and $l=\sum_b\ketbra{l_b}{l_b}$. Equation~\eqref{eq:gram} becomes $\sigma_M=\sum_{a,b}\ketbra{\phi_{ab}}{\phi_{ab}}$, where
\[
 \phi_{ab}(x)=\bra{l_b}M_x\ket{r_a}.
\]
Since $l$ and $r$ have full rank, both families of factor vectors span $\cB$, and the support of $\sigma_M$ is the span of all vectors $x\mapsto\bra uM_x\ket w$. It contains in particular $x\mapsto\tr(M_x)=\sum_i\bra{e_i}M_x\ket{e_i}$ for an orthonormal basis $(e_i)$, that is,
\begin{equation}\label{eq:block-support}
 \supp\ketbra{\Psi_k}{\Psi_k}\subseteq\supp\sigma_M.
\end{equation}

\emph{4. Reduce to one node and use the stabilized local state.}
By~\eqref{eq:support-order}, $\ketbra{\Psi_k}{\Psi_k}\leq C\sigma_M$ for some finite $C$, and the inequality survives the partial trace over every node except $(v,0)$. Taking $M\geq2s+1$ as well, this reads $\rho_k(v,s)\leq C\rho_\infty(v,s)$, and \cref{lem:injrad} gives $\rho_\infty(v,s)=\rho_m(v,s)$ for every $m\geq2s+1$. This is~\eqref{eq:supportlift}; no bound on the nilpotence index $D$ in terms of $s$ was needed. Finally, if an effect $F\geq0$ had $\tr(F\rho_m)=0$, the inequality would force $\tr(F\rho_k)=0$ too.
\end{proof}

\input{stacs/figures/mps}

\subsection{An exact obstruction to a smaller threshold}\label{sec:sharpness}
{The following example shows that the bound $m\geq2s+1$ cannot be weakened uniformly to $m\geq2s$.} Take the base with one node and one loop, so that $\cN_m=C_m$, give each node two port qubits $L,R$ initially in $\ket{00}$, attach them coherently as $R_j\leftrightarrow L_{j+1}$, use the swap link layer and the node unitaries
\begin{equation}\label{eq:witness-gates}
 U_1=\iSWAP(H\otimes H),\qquad U_2=\CZ(H\otimes H),
\end{equation}
and $\iSWAP\ket{01}=i\ket{10}$, $\iSWAP\ket{10}=i\ket{01}$, with $\ket{00},\ket{11}$ fixed. Here \(H\) is the standard Hadamard gate and \(\mathrm{CZ}=\operatorname{diag}(1,1,1,-1)\).

\begin{proposition}[Two-round sharpness witness]\label{prop:sharpness}
Let $\ket{a_\pm}=(\ket{00}\pm i\ket{11})/\sqrt2$. After the two rounds in~\eqref{eq:witness-gates},
\begin{equation}\label{eq:sharpness}
 \rho_{C_2}=\ketbra{a_+}{a_+},\qquad
 \rho_{C_4}=\ketbra{a_-}{a_-},\qquad
 \rho_{C_3}=\rho_{C_m}=\tfrac12\bigl(\ketbra{00}{00}+\ketbra{11}{11}\bigr)
 \quad(m\geq5).
\end{equation}
In particular $\supp\rho_{C_2}\not\subseteq\supp\rho_{C_4}$, although $C_4$ is a cover of $C_2$ and $4=2s$.
\end{proposition}
Appendix A gives an exact stabilizer proof. This example, though, says nothing about optimality at each horizon separately, but shows that local support need not be nested between small covers.

%% file: stacs/figures/mps.tex
\begin{figure}[t]
\centering
\begin{tikzpicture}[every node/.style={font=\small},
 box/.style={draw,thick,minimum width=8mm,minimum height=7mm,fill=black!8},
 boundarybox/.style={draw,thick,minimum width=6mm,minimum height=7mm,fill=black!18}]
 \foreach \i in {0,1,2} {
  \node[box] (a\i) at (1.2*\i,0) {$A$};
  \draw (a\i.north)--++(0,.4) node[above] {$x_{\i}$};
 }
 \draw[thick] (a0)--(a1)--(a2);
 \draw[thick] (a2.east)--++(.4,0)--++(0,-.7)--(-.8,-.7)--(-.8,0)--(a0.west);
 \node at (.9,-1.1) {Closed $k$-sheet network: $\tr(M_x)$};
 \begin{scope}[shift={(7,0)}]
  \node[boundarybox] (l) at (-1.1,0) {$l$};
  \foreach \i in {0,1,2} {
   \node[box] (b\i) at (1.2*\i,0) {$A$};
   \draw (b\i.north)--++(0,.4) node[above] {$x_{\i}$};
  }
  \node[boundarybox] (r) at (3.5,0) {$r$};
  \draw[thick] (l)--(b0)--(b1)--(b2)--(r);
  \node at (1.2,-1.1) {$k$ consecutive sheets in a larger network: $\tr(lM_xrM_{x'}^\dagger)$};
 \end{scope}
\end{tikzpicture}
\caption{The same sheet tensor generates every cyclic cover. Closing $k$ sheets gives $\Psi_k(x)=\tr(M_x)$. In a sufficiently large cover, contracting the exterior in the doubled tensor network gives positive-definite boundary operators $l,r$. The right drawing indicates these boundaries schematically; the density matrix also contains the conjugate physical indices $x'$. Its support contains the closed-chain vector.}
\label{fig:mps}
\end{figure}
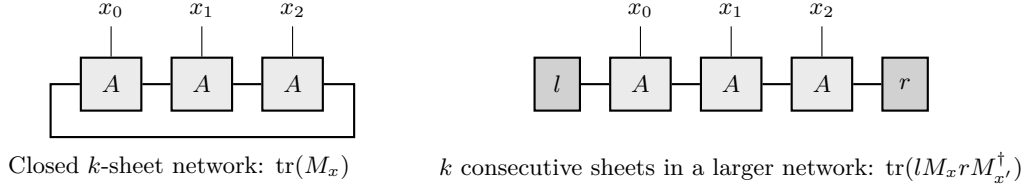

%% file: stacs/sections/05-universal.tex
\section{Universal covers and high-girth finite covers}\label{sec:universal}

The cyclic theorem reaches beyond cyclic towers once one passes through suitable finite covers. Let $U\to\cN$ be the universal covering tree, with all port, processor and input data lifted, and for a node $\tilde v$ above $v$ define $\rho_U(\tilde v,s)$ by the finite causal-cone prescription of \cref{sec:locality}. Girth counts a loop as a cycle of length one and a pair of parallel edges as a cycle of length two; a tree has infinite girth.

\begin{lemma}[Increasing girth by cyclic factors]\label{lem:homology}
Given a finite connected port graph $G$, an integer $p\geq2$ and a finite $h\geq0$, there is a finite connected cover $H\to G$ of girth greater than $h$ that is a composition of cyclic $p$-sheet covers. Every factor uses voltages in $\{-1,0,1\}$ and belongs to a connected cyclic tower.
\end{lemma}
\begin{proof}
If $G$ is a tree, take $H=G$, the empty composition. Otherwise choose a spanning tree and orient the $r=|E(G)|-|V(G)|+1$ edges outside it, the chords. Give the tree edges voltage zero in $\Z_p^r$ and the $i$th chord voltage $e_i$, the $i$th standard basis vector, with opposite ports assigned opposite voltages. The resulting cover is connected, because the fundamental cycle of the $i$th chord has voltage $e_i$ and the voltages of closed walks therefore generate the deck group $\Z_p^r$. 

This finite homology cover factors into $r$ successive cyclic $p$-sheet covers, one coordinate at a time. At factor $i$, every lift of the oriented $i$th chord has voltage $1$, its reverse has voltage $-1$, and every other link has voltage zero. The fundamental cycle of that chord uses only the chord and tree edges, so it lifts closed through the previous coordinates and has voltage $1$ at this factor, and the cyclic cover of this factor is connected at \emph{every} sheet count, not only at $p$. 

Let $g$ be the girth before the construction. A simple cycle in a cover projects to a non-backtracking closed walk, which contains a cycle of no greater length, so covering never decreases the girth. If the homology cover had a cycle of length $g$, its projection would be a shortest simple cycle of $G$, since a repeated vertex in the projection would give a shorter cycle. That cycle uses some chord exactly once, so its voltage has a coordinate $\pm1$ modulo $p$ and it cannot lift closed. The new girth is therefore strictly larger than $g$. 

Repeat the construction on the new graph as long as its girth is at most $h$. The girth grows by at least one each time, so finitely many repetitions suffice, and every factor is a finite connected cyclic $p$-sheet cover with the required voltages.
\end{proof}

\begin{theorem}[Support inclusion into the universal cover]\label{thm:universal}
Let $\cN$ be a finite connected port graph and consider a lift-uniform, product-prepared protocol that satisfies the finite-horizon resource condition and is defined on all finite covers of $\cN$. Then for every node $v$ and every finite horizon $s$,
\begin{equation}\label{eq:universal}
 \supp\rho_{\cN}(v,s)\subseteq\supp\rho_U(\tilde v,s).
\end{equation}
Consequently, for every finite connected covering $H\to\cN$ with $\girth(H)>2s$ and every node $w$ above $v$,
\begin{equation}\label{eq:highgirth}
 \rho_H(w,s)=\rho_U(\tilde v,s),\qquad
 \supp\rho_{\cN}(v,s)\subseteq\supp\rho_H(w,s).
\end{equation}
The local dynamics must be defined on the covers used; no task-class membership, correctness or termination on them is required.
\end{theorem}
\begin{proof}
Choose an integer $p\geq\max(2,2s+1)$. By \cref{lem:homology}, there is a finite sequence
\[
 G_0=\cN\longleftarrow G_1\longleftarrow\cdots\longleftarrow G_t=H_0
\]
of cyclic $p$-sheet covers ending at a graph of girth greater than $2s$. Lift all instance data along the sequence and choose corresponding nodes. At each factor the copied local rules define trace-preserving dynamics on the whole connected cyclic tower, not only on its $p$-sheet member, and the finite-horizon resource condition holds because every node type and every local rule comes from the fixed base. So \cref{thm:supportlift} applies with $k=1$ and $m=p$ at every factor, with no correctness assumption on the intermediate graphs, and gives
\[
 \supp\rho_{\cN}(v,s)\subseteq\supp\rho_{G_1}(v_1,s)
 \subseteq\cdots\subseteq\supp\rho_{H_0}(w_0,s).
\]
The map from the universal tree to $H_0$ is injective on the ball of radius $s$: two distinct vertices of that ball with the same image would give a non-backtracking closed walk of length at most $2s$, hence a cycle of length at most $2s$ in $H_0$, against its girth. By \cref{prop:locallift} and its finite-cone extension, $\rho_{H_0}(w_0,s)=\rho_U(\tilde v,s)$, which proves~\eqref{eq:universal}. The same injectivity argument applies to any finite cover $H$ of girth greater than $2s$, whose universal cover, with the lifted data and the corresponding root, is the same $U$; this proves~\eqref{eq:highgirth}.
\end{proof}

The injectivity radius that matters in~\eqref{eq:highgirth} is that of $U\to H$, not that of $H\to\cN$: when $H\to\cN$ is itself injective on a ball of radius $s$, \cref{prop:locallift} already gives equality with the base, and the inclusion is of interest precisely when it is not. It does not extend to every cover of small girth, as the witness $C_4\to C_2$ shows.

%% file: stacs/sections/06-consequences.tex
\section{Consequences for exact computation}\label{sec:consequences}

In $\Counting$ every node must output the total number of nodes; all nodes have the same input, which contains no size promise. In $\Election$ the outputs are $\mathsf{leader}$ and $\mathsf{follower}$, with exactly one leader once all outputs have been produced. Unlike the support theorems, the consequences below depend on task-class membership: the instance where an output is produced and the lifted instance where it contradicts correctness must both be admissible.

\begin{theorem}[Zero-error information on a cyclic tower]\label{thm:nocount}
Let $(\cN_m)_{m\geq1}$ be a $\cC$-admissible cyclic tower with lifted input data, and let a protocol satisfying the hypotheses of \cref{thm:supportlift} be zero-error for a function $f$ on $\cC$. If some node of $\cN_k$ outputs $f(\cN_k)$ with positive probability, then
\begin{equation}\label{eq:eventual}
 f(\cN_m)=f(\cN_k)\quad\text{for all sufficiently large }m.
\end{equation}
In particular, a zero-error counting algorithm produces no output at all, with probability one, on every member of the tower, and the same holds for a zero-error algorithm that outputs a finite upper bound on the number of nodes.
\end{theorem}
\begin{proof}
The event that a given node eventually outputs is a countable union of finite-time events, so if it has positive probability, the output is already present after some finite horizon $s$ with positive probability. Let $F$ be the projector onto the classical states that record $f(\cN_k)$ in the finite reachable local state space; then $\tr(F\rho_k(v,s))>0$, and \cref{thm:supportlift} gives $\tr(F\rho_m(v,s))>0$ for every $m\geq2s+1$. Every such $\cN_m$ belongs to $\cC$, so zero-error correctness there forces~\eqref{eq:eventual}. For counting, $f(\cN_m)=m|V(\cN)|$ is not eventually constant, so no node of a tower member outputs with positive probability, and a finite union over the nodes excludes any output there. For an upper bound, an output produced at a finite horizon with positive probability takes some particular finite value $B$ with positive probability, the reachable output set being finite at that horizon; lifting it to an $m$ with $m\geq2s+1$ and $m|V(\cN)|>B$ contradicts correctness.
\end{proof}

\begin{corollary}[Total functions and partial certificates]\label{cor:functions}
If $f(\cN_m)$ is not eventually constant in $m$, a zero-error protocol for $f$ produces no output, with probability one, on any tower member. If a zero-error protocol produces an output with positive probability on every member, then $f$ is constant on the whole tower; in particular, exact computation of $f$ at every node forces constancy on the tower.
\end{corollary}
\begin{proof}
The first assertion is the contrapositive of~\eqref{eq:eventual}. For the second, take outputs on two members $\cN_k$ and $\cN_\ell$: both values agree with the same tail of $(f(\cN_m))_m$, hence with each other.
\end{proof}

The distinction between exact computation and partial certificates is real. On rings, $f(C_n)=\mathbf1_{\{n\geq4\}}$ is not constant, yet it has a partial zero-error certificate: each node picks one of four colors at random and keeps forwarding the set of colors it has seen, and a node that has seen all four knows that $n\geq4$. This happens with positive probability on every ring with $n\geq4$, never on a shorter one, and nothing is output otherwise. The protocol does not compute $f$ exactly and does not contradict \cref{cor:functions}.

\begin{proposition}[A unique leader gives counting]\label{prop:count}
On finite connected instances of diameter $D$, suppose that a protocol almost surely produces an irrevocable $\mathsf{leader}$ output at some finite time and that no second node ever produces that output. Then it can be augmented so that every node counts exactly within $O(D+1)$ further synchronous rounds after the leader appears. More generally, almost-sure at-most-one-leader safety alone yields a zero-error augmentation, which terminates on the runs where a leader appears in finite time and stays silent on the runs where none does.
\end{proposition}
\begin{proof}
Use flooding and echo~\cite{segall1983distributed}. A node starts the procedure when it outputs $\mathsf{leader}$. The first arrival of a token selects an incoming port as parent, simultaneous arrivals being resolved by a fixed local order of the ports, and the node forwards the token on every other port. Every token, including those on loops and parallel links, is answered with a parent or non-parent acknowledgment, and the flooding builds a rooted spanning tree of depth at most $D$. Each node waits for an acknowledgment on every port where it sent a token, which tells it its children, then for a subtotal from each child; it sends one plus these subtotals to its parent, and the root broadcasts the final sum. The nodes carry out these classical operations whatever the output or stopping state of the original protocol, the classical slots carrying the auxiliary messages alongside its own. Flooding, acknowledgments, aggregation and broadcast take $O(D+1)$ rounds; with at most one initiator every count produced is correct, and without a leader the auxiliary protocol stays silent.

Let $T_L$ be the round of the first leader output, or $\infty$ if there is none. On the safety event, counting at all nodes is complete by $T_{\mathrm{count}}\leq T_L+c(D+1)$ whenever $T_L<\infty$, for an absolute constant $c$, so $\Pr(T_L<\infty)=1$ gives almost-sure finite counting. Neither $\mathbb E[T_L]<\infty$ nor earlier follower outputs are needed.
\end{proof}

\begin{corollary}[No exact leader election]\label{cor:noelection}
On a class $\cC$ containing a $\cC$-admissible cyclic tower there is no exact leader-election algorithm. More strongly, suppose an algorithm is safe on every instance of the class, meaning that with probability one no two nodes ever output $\mathsf{leader}$ and no completed output consists of $\mathsf{follower}$ alone. Then on every tower member the probability that some node outputs $\mathsf{leader}$ is zero, and so is the probability that all nodes output.
\end{corollary}
\begin{proof}
Compose the protocol with the leader-triggered procedure of \cref{prop:count}. Since at most one leader ever appears, the composition is zero-error for counting, and a leader appearing at a finite time with positive probability on a tower member would give a correct counting output there with positive probability, against \cref{thm:nocount}. With no leader and no valid all-follower completion, completion has probability zero. An exact election algorithm satisfies the safety conditions because its outputs are irrevocable and it terminates almost surely.
\end{proof}

\paragraph*{Which classes and which instances?}
The impossibility of exact solutions applies to all connected graphs, to connected graphs of maximum degree at most any fixed $\Delta\geq2$, to connected regular graphs of any fixed degree at least two, to connected bipartite graphs, to rings and to tori of fixed dimension. Each class contains a graph with a cycle and a connected cyclic tower preserving its defining property, and a simple base graph yields simple covers when the class requires simplicity. For tori, cutting between two consecutive coordinate hyperplanes gives
\[
 \Z_{n_1m}\times\Z_{n_2}\times\cdots\times\Z_{n_d},\qquad m\geq1.
\]
A zero-error protocol therefore cannot determine an unknown period on that tower, even when the other periods are fixed and known. 

The stronger silence statement is about the tower members, not about every graph of a class containing the tower. On the class of all connected graphs, a node of degree one whose only neighbor reports degree one can correctly output $2$, the network being $K_2$; this is a partial counting algorithm, not an exact solution on the whole class. Trees have no proper connected covers and are not excluded merely because towers exist elsewhere in the class. For rings, \cref{thm:nocount} answers the unknown-size and unknown-bound question of~\cite[Question 1]{aharonov2017dining} negatively, under the stated assumptions, including protocols with almost-sure finite termination and no deterministic time-bound. If an asynchronous algorithm must be correct under every admissible schedule and admits a fair synchronous schedule represented by this finite-horizon model, the same impossibility follows by restricting to that schedule; this reduction does not claim support lifting for arbitrary asynchronous schedules.

%% file: stacs/sections/07-discussion.tex
\section{Scope and further questions}\label{sec:discussion}

We have established a quantum covering principle for local supports: under our assumptions, every local outcome possible on a member of a connected cyclic tower remains possible on all sufficiently large members. Quantum statistics can still distinguish particular sizes, as the two-round witness shows. But no finite-time local observation can certify the size against arbitrarily large admissible members. Interference defeats the lifting of global records; zero-error lower bounds only need the lifting of local supports.

{Product preparation excludes classical correlations between nodes as well as entanglement:} the separable, permutation-invariant state $
 \omega_n=\frac1n\sum_{j=1}^{n}
 \ketbra{0\cdots1_j\cdots0}{0\cdots1_j\cdots0}$
yields exactly one leader as soon as each node measures its bit. Independent local memoryless noise with a finite-dimensional dilation is covered by the same proof, by keeping each fresh local environment in the dilation. {Environments initially correlated across different nodes and size-dependent shared initial resources that do not lift uniformly are outside our assumptions. Temporal memory retained within one node can instead be included in its active system, provided the enlarged nodes satisfy product preparation and the uniform finite-horizon resource bound.}

The universal-cover result gives a common upper bound on the supports of all finite quotients and equality of local states on covers of large girth, but a full characterization of exact anonymous quantum computation remains open.  For bounded-error computation, the relevant question is quantitative: how small can an outcome's probability become when it is transferred to a larger network?  The theorem does not provide a probability bound that applies to all protocols and horizons. Extending the support theorem to asynchronous dynamics remains an open question.

\paragraph*{Use of generative AI.}
AI tools were used for double checking proofs and improve drafting. Work with Fable (Claude) and Astra (ChatGPT) contributed to producing exploratory simulators, verification scripts, to scrutinize proofs and model assumptions, check literature comparisons, produce separate verification implementations, and revise and typeset this version. The original idea, all theorems, lemmas and proofs have been developed entirely by humans, for humans. Both authors have contributed equally.

%% file: stacs/sections/appendix.tex
\section{Exact certificate for the two-round witness}\label{app:sharpness}

We prove \cref{prop:sharpness} by hand, with no numerical rank test. Write $X,Y,Z$ for the Pauli matrices, with $Y=iXZ$, and $L_j,R_j$ for the two qubits at node $j$. The initial state is stabilized by all single-qubit $Z$ operators, and the gates are Clifford, since
\[
 \iSWAP=(S\otimes S)\CZ\SWAP,\qquad S=\operatorname{diag}(1,i).
\]
Conjugate the initial stabilizer generators through $U_1$, the first link layer, $U_2$ and the second link layer, in this order. For $n\geq3$, the resulting generators at node label $j$ are the pair
\begin{align}
 g_j&=-Z_{R_{j-2}}X_{L_j}Y_{R_j}Z_{L_{j+2}},\label{eq:stabilizer-g}\\
 h_j&=-Z_{R_{j-2}}Y_{L_j}X_{R_j}Z_{L_{j+2}}.\label{eq:stabilizer-h}
\end{align}
Indices are taken modulo $n$. Their product is
\begin{equation}\label{eq:local-zz}
 g_jh_j=Z_{L_j}Z_{R_j}.
\end{equation}
A stabilizer state on $N$ qubits has density matrix $2^{-N}\sum_{P\in\mathcal S}P$, and the partial trace kills every Pauli that acts nontrivially outside the retained node. The local state at node $j$ is therefore $\tfrac14\sum_{P\in\mathcal S_j}P$, where $\mathcal S_j$ is the subgroup of stabilizers supported entirely at node $j$.

\paragraph*{The case $n=4$.}
The exterior factors in~\eqref{eq:stabilizer-g} lie together at node $j+2$. Multiplying by its local $ZZ$ generator from~\eqref{eq:local-zz} produces $-X_{L_j}Y_{R_j}$; likewise~\eqref{eq:stabilizer-h} produces $-Y_{L_j}X_{R_j}$. These are two independent commuting local stabilizers, so the local state is pure and
\[
 \rho_{C_4}=\tfrac14(I-X\otimes Y-Y\otimes X+Z\otimes Z)
            =\ketbra{a_-}{a_-}.
\]

\paragraph*{The cases $n=3$ and $n\geq5$.}
Here $j-2$ and $j+2$ are distinct nodes, neither of them equal to $j$. For a product of generators to be supported at $j$ alone, the $X$ components at every other node $\ell$ force it to contain both of $g_\ell,h_\ell$ or neither, and containing both contributes only $Z_{L_\ell}Z_{R_\ell}$. Containing just one of $g_j,h_j$ leaves a single exterior $Z$ at each of two distinct nodes, which no product of local $ZZ$ operators can remove. Hence the only local stabilizers are $I$ and $Z_{L_j}Z_{R_j}$, and
\[
 \rho_{C_n}=\tfrac14(I+Z\otimes Z)
 =\tfrac12\bigl(\ketbra{00}{00}+\ketbra{11}{11}\bigr).
\]

\paragraph*{The case $n=2$.}
Direct conjugation, keeping the two parallel links as swaps of distinct port pairs, maps the four initial generators to
\[
 Y_{L_0}X_{R_0},\quad X_{L_0}Y_{R_0},\quad
 Y_{L_1}X_{R_1},\quad X_{L_1}Y_{R_1}.
\]
Thus
\[
 \rho_{C_2}=\tfrac14(I+X\otimes Y+Y\otimes X+Z\otimes Z)
            =\ketbra{a_+}{a_+}.
\]
Finally $\langle a_+|a_-\rangle=0$: the supports on $C_2$ and $C_4$ are orthogonal, which proves the non-inclusion and completes the certificate.

%% file: main.bbl
\begin{thebibliography}{10}

\bibitem{aharonov2017dining}
Dorit Aharonov, Maor Ganz, and Loick Magnin.
\newblock Dining philosophers, leader election and ring size problems, in the
  quantum setting.
\newblock {\em arXiv preprint arXiv:1707.01187}, 2017.
\newblock \href {https://doi.org/10.48550/arXiv.1707.01187}
  {\path{doi:10.48550/arXiv.1707.01187}}.

\bibitem{angluin1980local}
Dana Angluin.
\newblock Local and global properties in networks of processors.
\newblock In {\em Proceedings of the 12th Annual ACM Symposium on Theory of
  Computing (STOC)}, pages 82--93, 1980.

\bibitem{arrighi2019overview}
Pablo Arrighi.
\newblock An overview of quantum cellular automata.
\newblock {\em Natural Computing}, 18(4):885--899, 2019.
\newblock \href {https://doi.org/10.1007/s11047-019-09762-6}
  {\path{doi:10.1007/s11047-019-09762-6}}.

\bibitem{cirac2017mpu}
J.~Ignacio Cirac, David P{\'e}rez-Garc{\'i}a, Norbert Schuch, and Frank
  Verstraete.
\newblock Matrix product unitaries: Structure, symmetries, and topological
  invariants.
\newblock {\em Journal of Statistical Mechanics: Theory and Experiment},
  2017(8):083105, 2017.
\newblock \href {https://doi.org/10.1088/1742-5468/aa7e55}
  {\path{doi:10.1088/1742-5468/aa7e55}}.

\bibitem{coiteuxroy2024coloring}
Xavier Coiteux-Roy, Francesco d'Amore, Rishikesh Gajjala, Fabian Kuhn,
  Fran{\c{c}}ois Le~Gall, Henrik Lievonen, Augusto Modanese, Marc-Olivier
  Renou, Gustav Schmid, and Jukka Suomela.
\newblock No distributed quantum advantage for approximate graph coloring.
\newblock In {\em Proceedings of the 56th Annual ACM Symposium on Theory of
  Computing (STOC 2024)}, 2024.
\newblock arXiv:2307.09444.
\newblock URL: \url{https://arxiv.org/abs/2307.09444}.

\bibitem{fannes1992finitely}
Mark Fannes, Bruno Nachtergaele, and Reinhard~F. Werner.
\newblock Finitely correlated states on quantum spin chains.
\newblock {\em Communications in Mathematical Physics}, 144(3):443--490, 1992.
\newblock \href {https://doi.org/10.1007/BF02099178}
  {\path{doi:10.1007/BF02099178}}.

\bibitem{farrelly2020review}
Terry Farrelly.
\newblock A review of quantum cellular automata.
\newblock {\em Quantum}, 4:368, 2020.
\newblock \href {https://doi.org/10.22331/q-2020-11-30-368}
  {\path{doi:10.22331/q-2020-11-30-368}}.

\bibitem{itai1990symmetry}
Alon Itai and Michael Rodeh.
\newblock Symmetry breaking in distributed networks.
\newblock {\em Information and Computation}, 88(1):60--87, 1990.

\bibitem{linial1992locality}
Nathan Linial.
\newblock Locality in distributed graph algorithms.
\newblock {\em SIAM Journal on Computing}, 21(1):193--201, 1992.

\bibitem{peleg2000distributed}
David Peleg.
\newblock {\em Distributed Computing: A Locality-Sensitive Approach}.
\newblock SIAM, 2000.

\bibitem{perezgarcia2007mps}
David P{\'e}rez-Garc{\'\i}a, Frank Verstraete, Michael~M. Wolf, and J.~Ignacio
  Cirac.
\newblock Matrix product state representations.
\newblock {\em Quantum Information and Computation}, 7(5--6):401--430, 2007.
\newblock URL: \url{https://arxiv.org/abs/quant-ph/0608197}.

\bibitem{roget2024scheme}
Mathieu Roget and Giuseppe Di~Molfetta.
\newblock A quantum walk-based scheme for distributed searching on arbitrary
  graphs.
\newblock In {\em Cellular Automata Technology (ASCAT 2024)}, volume 2021 of
  {\em Communications in Computer and Information Science}, pages 72--83.
  Springer, 2024.
\newblock arXiv:2310.10451.
\newblock \href {https://doi.org/10.1007/978-3-031-56943-2_6}
  {\path{doi:10.1007/978-3-031-56943-2_6}}.

\bibitem{roget2026model}
Mathieu Roget and Giuseppe Di~Molfetta.
\newblock A quantum walk inspired model for distributed computing on arbitrary
  graphs.
\newblock {\em Natural Computing}, 25:24, 2026.
\newblock arXiv:2502.21232.
\newblock \href {https://doi.org/10.1007/s11047-026-10078-5}
  {\path{doi:10.1007/s11047-026-10078-5}}.

\bibitem{segall1983distributed}
Adrian Segall.
\newblock Distributed network protocols.
\newblock {\em IEEE Transactions on Information Theory}, 29(1):23--35, 1983.

\bibitem{tani2017solitude}
Seiichiro Tani.
\newblock A fast exact quantum algorithm for solitude verification.
\newblock {\em Quantum Information and Computation}, 17(1--2):15--40, 2017.
\newblock arXiv:1612.05317.
\newblock URL: \url{https://arxiv.org/abs/1612.05317}.

\bibitem{tani2012exact}
Seiichiro Tani, Hirotada Kobayashi, and Keiji Matsumoto.
\newblock Exact quantum algorithms for the leader election problem.
\newblock {\em ACM Transactions on Computation Theory}, 4(1):1:1--1:24, 2012.
\newblock URL: \url{https://arxiv.org/abs/0712.4213}, \href
  {https://doi.org/10.1145/2141938.2141939}
  {\path{doi:10.1145/2141938.2141939}}.

\bibitem{yamashita1996computing}
Masafumi Yamashita and Tsunehiko Kameda.
\newblock Computing on anonymous networks: Part {I}---characterizing the
  solvable cases.
\newblock {\em IEEE Transactions on Parallel and Distributed Systems},
  7(1):69--89, 1996.

\end{thebibliography}
